\documentclass[runningheads]{llncs}
\usepackage{amssymb}
\usepackage{graphicx}
\usepackage{amsmath}
\usepackage[ruled,linesnumbered,vlined]{algorithm2e}
\usepackage{epsfig}
\usepackage{url}

\urldef{\mailsa}\path|{vamsik,rajasek,hieu}@engr.uconn.edu|

\begin{document}

\title{An Efficient Algorithm For Chinese Postman Walk on Bi-directed de Bruijn Graphs} 

\titlerunning{Fast CPP Algorithms}

\author{Vamsi Kundeti \and Sanguthevar Rajasekaran \and Heiu Dinh}
\institute{Department of Computer Science and Engineering\\
University of Connecticut\\
Storrs, CT 06269, USA\\
\mailsa\\
} \maketitle

\begin{abstract}
Sequence assembly from short reads is an important problem in biology. It is known that solving the sequence assembly problem exactly on a bi-directed de Bruijn
graph or a string graph is intractable. 
However finding a Shortest Double stranded DNA string (SDDNA) containing all the $k$-long 
words in the reads seems to be a good heuristic to get close to the original genome. This problem is equivalent to finding a cyclic Chinese Postman (CP) walk on the underlying un-weighted 
bi-directed de Bruijn graph built from the reads. The Chinese Postman walk Problem (CPP) is solved by reducing it to a general bi-directed flow on this graph which runs in $O(|E|^2\log^2(|V|))$ time.  

\vspace{0.1in}

In this paper we show that the cyclic CPP on bi-directed graphs can be solved without reducing
it to bi-directed flow. We present a $\Theta(p(|V|+|E|)\log(|V|) + (d_{max}p)^3 )$ time
algorithm to solve the cyclic CPP on a weighted bi-directed de Bruijn graph, where 
$p = \max\{|\{v | d_{in}(v)-d_{out}(v) > 0\}|, |\{ v | d_{in}(v) - d_{out}(v) < 0\}|\}$
and $d_{max} = \max\{ |d_{in}(v)-d_{out}(v)\}$. Our algorithm performs asymptotically
better than the bi-directed flow algorithm when the number of {\em imbalanced} nodes 
$p$ is much less than the nodes in the bi-directed graph. From our experimental
results on various datasets, we have noticed that the value of $p/|V|$ lies between $0.08\%$ and $0.13\%$ with $95\%$ probability.

\vspace{0.1in}

Many practical bi-directed de Bruijn graphs do not have cyclic CP walks. In such cases
it is not clear how the bi-directed flow can be useful in identifying contigs. Our
algorithm can handle such situations and identify maximal bi-directed sub-graphs
that have CP walks. A $\Theta(p(|V|+|E|))$ time heuristic algorithm based on these ideas
has been implemented for the SDDNA problem. This algorithm was tested on short reads
from a plant genome and achieves an approximation ratio of at most $1.0134$. We also present a $\Theta((|V|+|E|)\log(V))$ time algorithm for the single 
source shortest path problem on bi-directed de Bruijn graphs, which may be of 
independent interest. 
\end{abstract}

\section{Introduction}

Sequencing the human genome was one of the major scientific breakthroughs in the 
last seven years. Analysis of the sequenced genome can give us vital information about the
expression of genes, which in turn can help scientists to develop drugs for diseases. Thus
sequencing the genome of an organism is of fundamental importance in both medicine and biology.
Unfortunately the technology used in major human genome sequencing projects -- Human Genome
Project (HGP)~\cite{hgp_project} and Celera~\cite{celera_project}, was too expensive to be 
adopted in a large scale. This led to the research on {\em next-generation sequencing} 
methods. {\em Pyro-sequencing} technologies such as SOLiD, 454 and Solexa generate a large number of
short reads which are very accurate. Directed de Bruijn graph based sequence assembly 
algorithms such as \cite{velvet08} and \cite{pevzner01} seem to handle these short read data efficiently 
compared to the string graph based algorithms (see e.g., \cite{myers05}). Unfortunately solving the sequence 
assembly problem exactly on both these graph models seems intractable~\cite{bidirected_graph}. 
However heuristics such as finding a shortest string which includes all the $k$-mers (sub strings 
of length $k$) seem to yield results close to the original genome. In the case of directed de Bruijn graphs finding an Eulerian tour seems to yield good results. If the graph is not Eulerian then a Chinese Postman (CP) tour has been suggested in~\cite{pevzner01}.
 To account for the double strandedness of the DNA molecule we need to simultaneously search for two complimentary CP tours. 
In~\cite{bidirected_graph} the directed de Bruijn graphs are replaced with bi-directed de Bruijn 
graphs to find two complimentary CP tours simultaneously. A CP tour on the {\em un-weighted} 
bi-directed graph constructed from the reads serves as a solution to the {\em Shortest Double Stranded DNA}  string (SDDNA) problem. The solution presented 
in~\cite{bidirected_graph} solves the SDDNA problem by reducing it to a general weighted bi-directed 
flow problem. This algorithm runs in $O(|E|^2\log^2(V))$ time. 

In this paper we present algorithms for SDDNA/CPP  on bi-directed de Bruijn graphs without
using a bi-directed flow algorithm. Our algorithms are based on identifying shortest bi-directed
paths and use of weighted bi-partite matching. Our algorithms perform asymptotically better than the bi-directed flow algorithm when the {\em imbalanced} nodes in the bi-directed graphs are much smaller in number than $|V|$. This restriction seems to be true in practice from what we have observed in our experiments. On 
the other hand it turns out that in many practical situations these bi-directed de Bruijn graphs 
fail to have {\em cyclic} CP tours. In these cases it is not clear how the bi-directed flow 
algorithm~\cite{bidirected_graph} can help us in identifying a set of {\em contigs} covering every 
$k$-long word at least once. In contrast to this flow algorithm, our algorithm can be useful
in obtaining a {\em minimal} set of contigs when a {\em cyclic} CP tour does no exist. We now
summarize our results as follows. Firstly our deterministic algorithm to solve the {\em cyclic}
CPP on a general bi-directed graph takes $\Theta(p(|V|+|E|)\log(|V|) + (d_{max}\,p)^3)$ time, 
where $d_{max} = \max\{|d_{in}(v)-d_{out}(v)|, v\in V\}$, $p = \max\{|V^+|, |V^-|\}$, 
$V^+ = \{v | v\in V, d_{in}(v)-d_{out} > 0\}$ and $V^- = \{v | v\in V, d_{in}(v) - d_{out} < 0\}$.
Secondly we solve the SDDNA problem on an un-weighted bi-directed de Bruijn graph deterministically 
in $\Theta(p(|V|+|E|) +(d_{max}\,p)^3)$ time. As a consequence we also present a $\Theta((|V|+|E|)\log(V))$ 
time single source shortest bi-directed path algorithm, which may be of independent interest to 
some assembly algorithms such as Velvet~\cite{velvet08} -- TourBus heuristic.

The organization of the paper is as follows. In Section~\ref{sec:prelim} we provide some preliminaries. 
Section~\ref{sec:prob-def} defines the CPP and SDDNA problems. In Section~\ref{sec:bi-short-path} we
introduce our algorithm for single source shortest bi-directed paths, which is used as a component
in our main algorithm. The main algorithm is introduced in Section~\ref{sec:det_algo} along with
algorithms for several sub-problems. Section~\ref{sec:no-cpp} briefly explains how we can handle situations when
the bi-directed graphs do not have cyclic CP tours. A greedy algorithm that runs in $\Theta(p(|V|+|E|)$ time
is described in Section~\ref{sec:gdy_algo}. Finally experimental studies are reported in 
Section~\ref{sec:exp}.

\section{Preliminaries}
\label{sec:prelim}
Let $s \in \Sigma^{n}$ be a string of length $n$. Any substring $s_j$ (i.e., $s[j,\ldots j+k-1], n-k+1\geq j\geq 1$) of 
length $k$ is called a $k-$mer of $s$. The set of all $k-$mer's of a given string $s$ is called the $k-$spectrum of 
$s$ and is denoted by $\mathbb{S}(s,k)$. Given a $k-$mer $s_j$, $\bar{s_j}$ denotes the {\em reverse compliment} of $s_j$ (e.g., if $s_j = AAGTA$ then $\bar{s_j} = TACTT$). Let $\leq$ be the partial ordering among the strings of equal length, 
then $s_i \leq s_j$ indicates that string $s_i$ is lexicographically smaller than $s_j$. Given any $k-$mer $s_i$, 
let $\hat{s_i}$ be the lexicographically smaller string between $s_i$ and $\bar{s_i}$. 
We call $\hat{s_i}$ the {\em canonical} $k-$mer of $s_i$. More formally, if $s_i \leq \bar{s_i}$ then $\hat{s_i} = s_i$ 
else $\hat{s_i} = \bar{s_i}$. A $k-$molecule of a given $k-$mer $s_i$ is a tuple $(s_i,\bar{s_i})$ 
consisting of $s_i$ and its reverse compliment $\bar{s_i}$, the first entry in this tuple is called the 
positive strand and the second entry is called the negative strand.

A {\em bi-directed} graph is a generalized version of a standard directed graph. In a directed graph every 
edge ($\text{--}\rhd$ or $\lhd\text{--}$) has only one arrow head. On the other hand, in a bi-directed graph 
every edge ($\lhd\text{--}\rhd$, $\lhd\text{--}\lhd$,$\rhd\text{--}\lhd$ or $\rhd\text{--}\rhd$) has two arrow 
heads attached to it. Formally, let $V$ be the set of vertices of a bi-directed graph, $E = \{(v_i,v_j,o_1,o_2) | v_i,v_j\in V \wedge o_1,o_2\in\{\lhd,\rhd\}\}$ is the set of bi-directed edges in a bi-directed graph $G(V,E)$. 
A {\em walk} $w(v_i,v_j)$ between two nodes $v_i,v_j \in V$ of a bi-directed graph $G(V,E)$ is a
sequence $v_i,e_{i_1},v_{i_1},e_{i_2},v_{i_2}\ldots v_{i_m},e_{i_{m+1}},v_j$, such that for 
every intermediate vertex $v_{i_l},1\leq l \leq m$, the orientation of the arrow heads on either side is opposite. To make this more clear let $e_{i_l},v_{i_l},e_{i_{l+1}}$ be the sub-sequence in the walk $w(v_i,v_j)$,
$e_{i_l}=(v_{i_{l-1}},v_{i_l},o_1,o_2), e_{i_{l+1}} = (v_{i_{l}},v_{i_{l+1}},o_1,o_2)$
then for the walk to be valid $e_{i_l}.o_2 = e_{i_{l+1}}.o_1$. If $v_j = v_i$ and $e_{i_1}.o_1 = 
e_{i_{m+1}}.o_2$ then the walk is called {\em cyclic}. A walk on the bi-directed graph is referred to as
a {\em bi-directed walk}. We define a orientation function ${\mathcal O} : V^2 \rightarrow \{\rhd, \lhd\}^2$ which
gives the orientation of the bi-directed edge between a pair of vertices -- if one exists . 
For instance if $(v_i,v_j,\lhd,\rhd)$ is a bi-directed edge between $v_i$ and $v_j$ then 
${\mathcal O}(v_i,v_j) = \lhd-\rhd$. An edge which is adjacent on a vertex with a orientation 
$\rhd$ ($\lhd$) is called an {\em incoming} ({\em outgoing}) edge. The incoming(outgoing) degree 
of a vertex $v$ is denoted by $d_{in}(v)$ ($d_{out}(v)$). A vertex $v$ is called {\em balanced}
iff $d_{in}(v)-d_{out}(v)=0$. A vertex is called {\em imbalanced} 
iff $|d_{in}(v)-d_{out}(v)| >0$. The imbalance of a vertex is called {\em positive} 
iff $d_{in}(v)-d_{out}(v) > 0$. Similarly a vertex is {\em negative} imbalanced iff 
$d_{in}(v)-d_{out}(v)<0$. A bi-directed graph is called {\em connected} iff every pair of
vertices have a bi-directed walk between them.

A de Bruijn graph $D^k(s)$ of the order $k$ on a given string $s$ is defined as follows. The vertex set 
$V$ of $D^k(s)$ is defined as the $k-$spectrum of $s$ (i.e., $V = \mathbb{S}(s,k)$). We use the notation 
$suf(v_i,l)$($pre(v_i,l)$) to denote the suffix(prefix) of length $l$ in string $v_i$. The symbol $.$ denotes 
concatenation between two strings. Finally the set of directed edges $E$ of $D^k(s)$ is defined as follows 
$E = \{(v_i,v_j) | suf(v_i,k-1)=pre(v_j,k-1) \wedge v_i[1].suf(v_i,k-1).v_j[k] \in \mathbb{S}(s,k+1) \}$. 
We can further generalize the definition of a de Bruijn graph $B^k(S)$ on a set $S=\{s_1,s_2\ldots s_n\}$ 
of strings, $V = \displaystyle\cup_{i=1}^{n} \mathbb{S}(s_i,k)$ and $E = \{(v_i,v_j) | suf(v_i,k-1) = pre(v_j,k-1) \wedge \exists \, l : v_i[1].suf(v_i,k-1).v_j[k]\in \mathbb{S}(s_l,k+1)\}$.

To model the double strandedness of the DNA molecules we should also consider the reverse 
compliments ($\bar{S} =\{\bar{s_1},\bar{s_2}\ldots \bar{s_n}\})$ while we build the de Bruijn graph. 
To address this a bi-directed de Bruijn graph $BD^k(S\cup \bar{S})$ has been suggested in 
~\cite{bidirected_graph}. The set of vertices $V$ of $BD^k(S\cup \bar{S})$ consists of all the 
possible $k-$molecules from $\Sigma^k$. For every $k+1-$mer $z \in S\cup \bar{S}$, if $x,y$ are 
the two $k-$mer's of $z$ then an edge is introduced between the $k-$molecules ($v_i,v_j$) corresponding 
to $x$ and $y$. The orientations of the arrow heads on the edges is chosen as follows. If both $x,y$ 
are the positive strands in $v_i,v_j$ an edge $(v_i,v_j,\rhd,\rhd)$ is introduced. If $x$ is a positive 
strand in $v_i$ and $y$ is a negative strand in $v_j$ an edge $(v_i,v_j,\rhd,\lhd)$ is introduced. 
Finally if $x$ is a negative strand in $v_i$ and $y$ is a positive strand in $v_j$ an 
edge $(v_i,v_j,\lhd,\rhd)$ is introduced. 

\section{Problem Definitions}
\label{sec:prob-def}
A {\em Chinese Postman walk} in a bi-directed graph is a bi-directed walk which visits every 
edge at least once. A {\em cyclic Chinese Postman walk} of minimum cost on a weighted bi-directed graph is denoted as CPW. The problem of finding a CPW is referred to as CPP. The problem of finding a CPW on an un-weighted 
bi-directed de Bruijn graph (of order $k$) constructed from a set of reads is called the
{\em Shortest Double stranded DNA} string (SDDNA) problem.  In this paper we give algorithms for the cyclic CPP and SDDNA problems.

\section{Single Source Shortest Path Algorithm on a Bi-directed de Bruijn Graph}
\label{sec:bi-short-path}
We first present an algorithm for the single source shortest path problem on a 
bi-directed de Bruijn graph. The bi-directed de Bruijn graph in the context of 
sequence assembly has non-negative weights on the edges. This makes it possible to 
extend the classic Dijkstra's single source shortest path algorithm to these
graphs. In our algorithm we attach two labels for each vertex in the bi-directed 
graph. Given a source vertex $s$, the algorithm initializes all the labels
similar to Dijkstra's algorithm. In each stage of the algorithm a label with 
the smallest cost is picked and some of labels corresponding to adjacent
nodes are updated. The only major difference between Dijkstra's algorithm 
and our algorithm is the way we update the labels. Dijkstra's algorithm 
updates all the labels/nodes which are adjacent to the smallest label/node
currently picked. However our algorithm updates only those labels/nodes which
are consistent with the bi-directed walk property.  

We now give details of our algorithm and prove its correctness. Let 
$G=(V,E)$ be the bi-directed graph of interest. Also let $s$ be the source 
and $t$ be the destination. We are interested in finding a {\em shortest 
bi-directed} walk from $s$ to $t$. We introduce two labels $dist^+[u]$, 
$dist^-[u]$ for every vertex $u \in V$. The algorithm first initializes
labels corresponding to the source (i.e. $dist^+[s]$ and $dist^-[s]$) 
to zero. Along with this labels of all the nodes adjacent to $s$ are
also initialized with the corresponding edge weight. The orientation
of the edge determines the label we use for initialization. For instance,
if $(s,v)$ is a bi-directed edge with $\rhd-\lhd$ as the orientation,
the label $dist^-[v]$ is initialized with $w_{s,v}$. On the 
other hand $dist^-[v]$ is left uninitialized. In contrast, if
the orientation of the edge is $\rhd-\rhd$ then $dist^[v]$ is initialized
to $w_{s,v}$ and $dist^-[v]$ is left uninitialized. All the uninitialized 
labels contains $\infty$ by default. 

In each iteration of the algorithm a label with the minimum cost is picked. 
Since we have two types of labels, the minimum label can come from either 
$dist^+$ or $dist^-$. In the first case let $u^+$ be the node corresponding
to the minimum label during the iteration. This intuitively means that we 
have a path from $s$ to $u^+$ and the orientation of the edge adjacent to 
$u^+$ in this path is either $\lhd-\rhd$ or $\rhd-\rhd$ -- we are going to 
prove this fact later in the correctness. On the other hand if $u^+$ is 
different from the destination $t$, then $u^+$ may possibly appear as an 
internal node in the shortest bi-directed walk between $s$ and $t$. In this 
case the path through $u^+$ should satisfy {\em bi-directed walk constraint}. 
Thus we should explore only those node(s) adjacent to $u^+$ with an edge(s) 
orientated as $\rhd-\lhd$ or $\rhd-\rhd$. The orientation 
of the edge determines the type of the label we need to update -- similar to 
the label initialization. For instance let $(u^+,v)$ be an edge adjacent on 
$u^+$ with a orientation of $\rhd-\lhd$. In this case we should use label 
$dist^-[v]$ to make an update. Similarly if the orientation of the same edge 
is $\rhd-\rhd$ then $dist^+[v]$ is used in the update process. Consistent
with the classical terminology of the Dijkstra's algorithm, we refer to the 
minimum cost label picked in each iteration as the {\em permanent label}.
For instance if a label $dist^-[v]$ is picked to be the minimum label in an
iteration then we call $dist[v]$ as the {\em permanent label} of node $v$.
Now to prove the correctness of the algorithm. It is sufficient to show that the
cost on the permanent label of a node in each iteration is the weight of the shortest
bi-directed path from $s$ to that node.

\restylealgo{ruled}\linesnumbered
\begin{algorithm}
\SetKwInOut{Input}{INPUT}
\SetKwInOut{Output}{OUTPUT} 
\Input{Bi-directed graph $G=(V,E)$ and two vertices $s,t \in V$}
\Output{Cost of the shortest bi-directed path between $s$ and $t$}
$ $\\
$dist^{+}[s] = dist^{-}[s] = 0$ \\
$dist^{+}[v] = dist^{-}[v] = \infty \,\,\forall v\in V \wedge v\neq s$ \\
$ $\\
\While{ $dist^{+} \neq \phi$ or $dist^{-} \neq \phi$}{
	$u^+ = \min_{u}\{dist^+\}$ \\
	$u^- = \min_{u}\{dist^-\}$ \\
$ $\\
	\If{$u^+ = t$ or $u^- = t$}{
		{\sf return} $\min\{dist^+[u^+], dist^-[u^-]\}$
	}
$ $\\
$ $ \\
	\If{$dist^+[u^+] < dist^-[u^-]$}{
		$U^+ = \{ v | (u^+,v)\in E \wedge ({\mathcal O})(u^+, v) = \lhd-\lhd)\}\}$ \\
		$U^- = \{ v | (u^+,v)\in E \wedge ({\mathcal O})(u^+, v) = \lhd-\rhd)\}\}$ \\
		$dist[u^+] = dist^+[u^+]$ \\
		$dist^+ = dist^+-\{u^+\}$ \\
	}\Else{
		$U^+ = \{ v | (u^-,v)\in E \wedge ({\mathcal O})(u^-, v) = \rhd-\lhd)\}\}$ \\
		$U^- = \{ v | (u^-,v)\in E \wedge ({\mathcal O})(u^+, v) = \rhd-\rhd)\}\}$ \\
		$dist[u^-] = dist^-[u^-]$ \\
		$dist^- = dist^--\{u^-\}$ \\
	}
$ $\\
$ $ \\
	\ForEach{$u \in dist^+$}{
		$dist^+[u] = \min\{ dist^+[u], dist^+[u^+]+w[u^+, u]\} $\\
	}
	\ForEach{$u \in dist^-$}{
		$dist^-[u] = \min\{ dist^-[u], dist^-[u^-]+w[u^-, u]\} $\\
	}
$ $\\
}
$ $ \\
{\sf return} $\infty$

\caption{Algorithm to find the shortest bi-directed path from $s$ to $t$ }
\label{algo:shortest_path}
\end{algorithm}

\begin{figure*}
\label{fig:bi-walk-example}
\begin{center}
\includegraphics[scale=0.7]{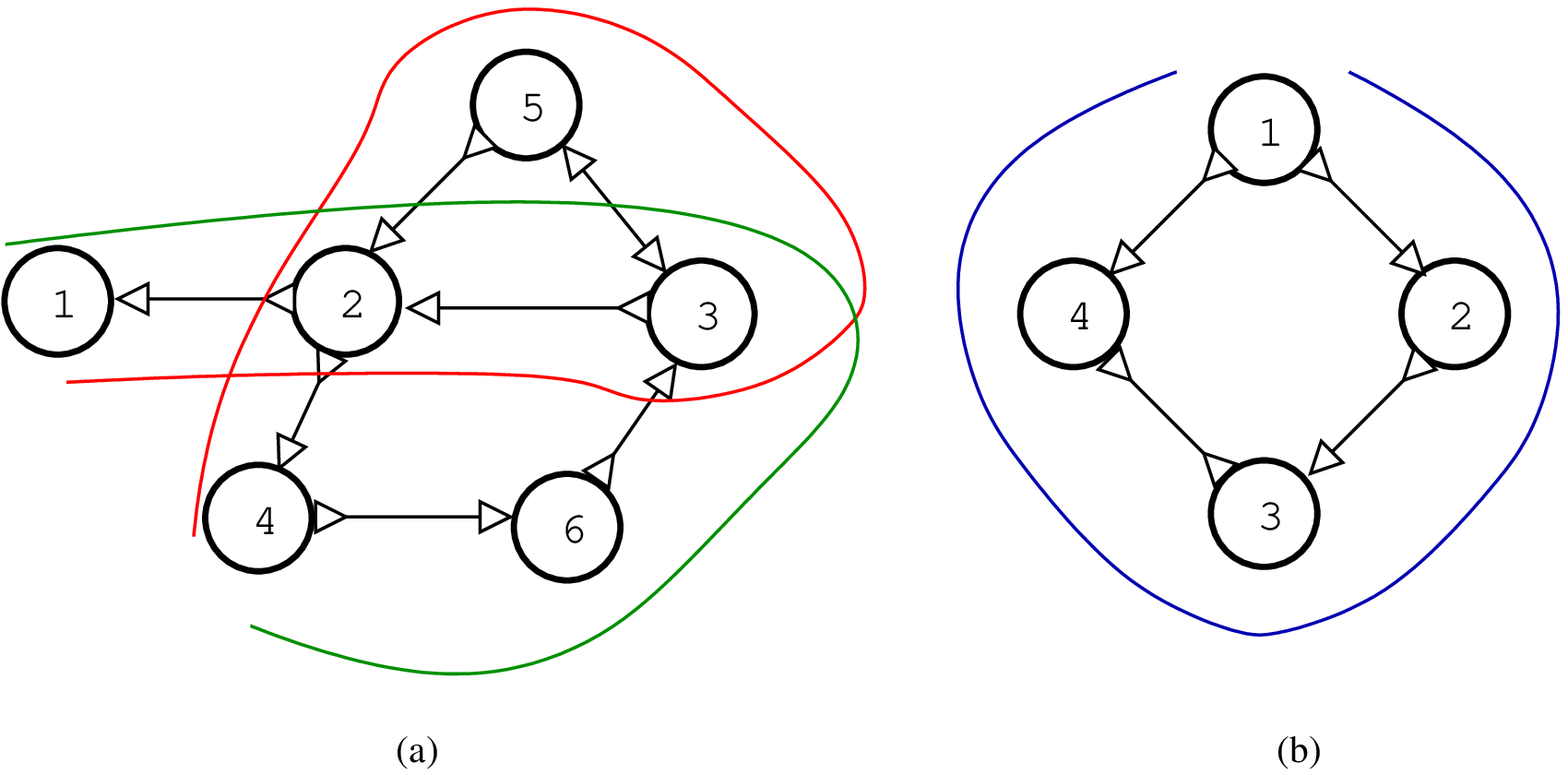}
\end{center}
\caption{{\bf(a)} node $4$ contains two bi-directed walks from node $1$, the green colored path is the shortest.{\bf (b)} the walk starting from node $1$ and ending at node $1$ is a Chinese walk but not a cyclic Chinese walk.}
\end{figure*}

\begin{theorem}
The permanent label of a node $u\in V$ in each iteration of {\bf Algorithm~\ref{algo:shortest_path}} is the weight of the
shortest bi-directed path from $s$ to $u$.
\end{theorem}
\begin{proof}
We prove the statement by induction on the number ($n$) of iterations in 
Algorithm~\ref{algo:shortest_path}. We now prove the base case when $n=1$. Since
we have initialized $dist^+[s] = dist^-[s] = 0$ and the values of the remaining
both initialized and uninitialized nodes are $>0$; the first iteration picks $s$ and zero
is trivially the cost of shortest bi-directed path form $s$ to $s$. 

Assume that the statement is true for $n=1\dots k$. As per the induction hypothesis the permanent labels 
$dist[s], dist[v_{i_2}] \ldots dist[v_{i_k}]$ correspond to the costs of the shortest 
bi-directed paths between $s$ and $s,v_{i_2}\ldots v_{i_k}$. 

Now let $dist'[v_{i_{k+1}}] < dist[v_{i_{k+1}}]$
be the cost of the shortest bi-directed walk from $s$ to $v_{i_{k+1}}$. Also let 
$s,v_{j_2}\ldots v_{j_k},v_{i_{k+1}}$ be the path corresponding to the cost $dist'[v_{i_{k+1}}]$.
Note that $v_{j_k}$ cannot be one of the nodes with a permanent label. 
If not, we would have $dist'[v_{i_{k+1}}] = dist[v_{i_{k+1}}]$ 
(because we should have updated $v_{k+1}$ when the $v_{j_k}$ was given a permanent label) which is a contradiction. 
Now let $dist'[v_{j_k}]$ be the cost of the shortest path from $s$ to $v_{j_k}$. Clearly, 
$dist'[v_{j_k}] < dist[v_{i_{k+1}}]$ and this means that none of the nodes $v_{j_2}, v_{j_3} \ldots v_{j_k}$ haa a permanent label. Since in the iteration $n=1$ the algorithm updated the labels adjacent to
all the nodes this means that either $dist^+[v_{j_2}]$ or $dist^-[v_{j_2}]$ should have a cost $0<w_{s,j_2}$
and $dist'[v_{i_{k+1}}] \geq w_{s,j_2}$. In each iteration from $n = 1,\ldots, (k+1)$ we picked the globally 
minimum label $dist[v_{i_{k+1}}] < w_{s, j_2} \leq dist'[v_{i_{k+1}}]$ which is a contradiction. 
$\Box$
\end{proof}

We now give a simple example to illustrate the algorithm. Consider the bi-directed graph in 
Figure~\ref{fig:bi-walk-example}(a), with a unit weight on every edge. Let $s=1$ and $t=4$ for instance. 
From Figure~\ref{fig:bi-walk-example}(a) we see two bi-directed walks -- {\em red}, {\em green}.
The {\em green} path is the shortest path of length $4$ units. Now let us run our algorithm on this
graph. The algorithm starts with initializing from $s=1$ and from then on Table~\ref{tab:bi-example}
shows how the $dist$ labels are updated until iteration $7$ where we reach the target node $4$ and stop
the algorithm.

\section{Terminal Oriented Shortest Bi-directed Walks}
In the previous section we have seen how to find a shortest bi-directed walk between two nodes in 
a given bi-directed graph. We now define a {\em terminal oriented bi-directed walk} as follows. 
Let $w(v_i,v_j) = v_i,e_{i_1},v_{i_1},e_{i_2},v_{i_2}\ldots v_{i_m},e_{i_{m+1}},v_j$ be any 
bi-directed walk between two nodes $v_i$ and $v_j$ in a bi-directed graph. Then this bi-directed
walk $w(v_i,v_j)$ is called {\em terminal oriented bi-directed walk} iff $e_{i_1}.o_1 = \rhd$ and
$e_{i_{m_1}}.o_2 = \rhd$. For example in Figure~\ref{fig:bi-walk-example}(a) there are two bi-directed
walks between nodes $4$ and $1$ -- marked with green and red. However only the green bi-directed walk is terminally oriented. A terminal oriented bi-directed walk
$w$ is called the {\em shortest terminal oriented bi-directed walk} iff there is no other terminal
oriented bi-directed walk shorter than $w$. 

\subsection{An algorithm for finding a terminal oriented shortest bi-directed walk}
It is easy to modify Algorithm~\ref{algo:shortest_path} to find a terminal oriented shortest path
between $s$ and $t$. We only have to modify the initialization step and the step which checks
if the target node has been reached. During the initialization at line $2$ of Algorithm~\ref{algo:shortest_path}
we make $dist^+[s]=0$ and $dist^-[s]=\infty$. This avoids the exploration of bi-directed walks which does
not start with $\rhd$. In line $9$, we stop our exploration only if $u^+=t$. These changes
ensure that the bi-directed walk at $s$ starts with $\rhd$ and ends with $\rhd$ at $t$. 

\section{A Sufficient Condition for an Eulerian Tour on a Bi-directed Graph} 
The following Lemma~\ref{lem:etour}~\cite{bidirected_graph} is a sufficient condition for a cyclic 
Eulerian tour in a bi-directed graph. A bi-directed graph which has a cyclic Eulerian tour is called
an Eulerian bi-directed graph.
\begin{lemma}
\label{lem:etour}
A connected bi-directed graph is Eulerian if and only if every vertex is balanced.
\end{lemma}
Note that if a bi-directed graph is Eulerian then a cyclic CP walk is the same as a cyclic Eulerian
walk. We emphasize the cyclic adjective for the following reason. Figure~\ref{fig:bi-walk-example}(b) has a 
CP walk starting and ending at vertex $1$. However the CP walk is not cyclic because the walk starts
with $\rhd$ and ends with $\lhd$. The bi-directed graph in Figure~\ref{fig:bi-walk-example}(b)
is not balanced. If the bi-directed graph is not Eulerian, the key strategy to find
a cyclic CP walk is to make it Eulerian by introducing multi-edges into the original graph. The hope is that 
introducing multi-edges would make the bi-directed graph balanced. Thus a cyclic Eulerian walk on a balanced 
multi-edge bi-directed graph would give a cyclic CP walk on the original graph. Since we are interested
in finding a shortest cyclic CP walk, we would like to minimize the number of multi-edges we introduce
in the original graph.

\begin{figure*}
\label{fig:multi-example}
\begin{center}
\includegraphics[scale=0.65]{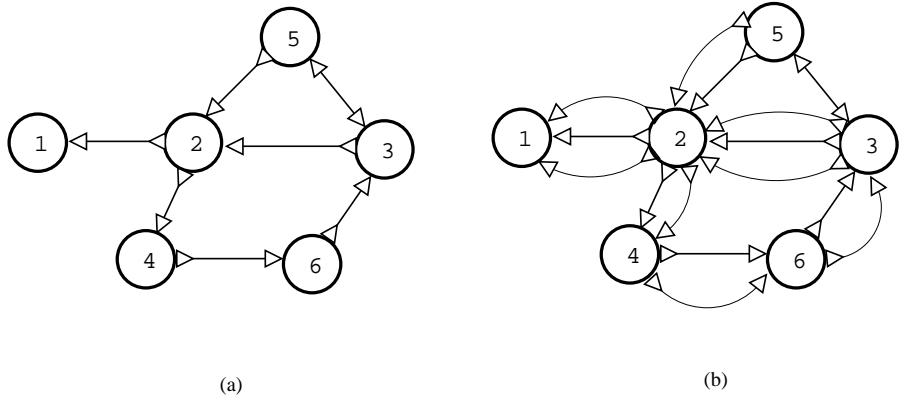}
\end{center}
\caption{{\bf (a)} a simple bi-directed graph, {\bf (b)} a multi-bi-directed graph. Notice that orientations
of the multi-edges is the same as the orientation of the original edge.}
\end{figure*}

\section{A Deterministic Algorithm to Find a Cyclic CP Walk on a Bi-directed Graph}
\label{sec:det_algo}
We now describe our deterministic algorithm to find a cyclic CP walk on a weighted bi-directed 
graph. First we define a {\em multi-bi-directed graph} as a bi-directed graph in which an edge between two
nodes is overlaid at least once, without changing its orientation. Figure~\ref{fig:multi-example}(a)
shows a bi-directed graph; Figure~\ref{fig:multi-example}(b) shows a valid multi-bi-directed graph. Notice
that while overlaying the edge we did not change its orientation. Since the orientation of the multi-edges
is same as the original edges, any bi-directed walk involving multi-edges is consistent with the bi-directed
walk in the original graph. Another important property of the multi-bi-directed graphs is their ability to make
the nodes balanced. Notice that the vertex $3$ in the original bi-directed graph is positively imbalanced --
$d_{in}(v_3) = 2, d_{out}(v_3) = 1$. However in the multi-bi-directed graph in Figure~\ref{fig:multi-example}(b)
we are able to balance vertex $3$ by introducing some multi-edges into the original graph. Given a bi-directed
graph $G=(V,E)$, let $G^m=(V,E^m)$ be some multi-bi-directed graph corresponding to $G$. The following 
Lemma~\ref{lem:multi-lemma} gives a characterization for $G$ to have a cyclic CP walk. 
\begin{lemma}
\label{lem:multi-lemma}
A non Eulerian bi-directed graph $G=(V,E)$ has a cyclic Chinese Postman walk $\iff \exists$ a corresponding 
multi-bi-directed graph $G^m=(V,E^m)$ which is Eulerian.
\end{lemma}
\begin{proof}
If $G$ has a cyclic Chinese Postman walk, introduce a unique multi-edge in $G^m$ for every repeated edge in the 
cyclic Chinese Postman walk. This makes the cyclic Chinese Postman walk on the original graph a cyclic Eulerain walk on the
multi-bi-directed graph $G^m$. Proving the other direction is very similar. $\Box$
\end{proof}

Given a multi-bi-directed graph $G^m(V,E^m)$ corresponding to some bi-directed graph $G=(V,E)$, we define the 
{\em multi-bi-directed graph weight} as ${\cal W}(G^m) = \displaystyle\sum_{e\in E^m} c(e)$, where $c : e\in E \rightarrow \mathbb{R^+}$
is a cost function on the bi-directed graph $G(V,E)$.  We denote $G^*(V,E^*)$ as the minimum weight Eulerian multi-bi-directed graph 
corresponding to $G(V,E)$ if at all one exists. The following Lemmas are easy to prove.
\begin{lemma}
\label{lem:reformulate}
Finding a cyclic CP walk on a bi-directed graph $G(V,E)$ is equivalent to finding a minimum weight Eulerian multi-bi-directed graph $G^*(V,E^*)$ corresponding to $G$. 
\end{lemma}
\begin{lemma}
\label{lem:cost}
If a bi-directed-graph $G(V,E)$ has a cyclic CP walk then the cost of that walk is equal to the weight of $G^*(V,E^*)$.
\end{lemma}

\subsection{Balancing bi-partite graph}
\label{sec:bal-bi-graph}
Given a bi-directed de Bruijn graph $G(V,E)$ we define a corresponding {\em Balancing Bi-partite Graph},
$B(P,Q,E^b)$ as follows. Let $V^+=\{ v |\,d_{in}(v) - d_{out}(v) > 0\}$, 
$V^- = \{ v |\,d_{in}(v) - d_{out}(v) < 0\}$. 
$P = \cup_{p\in V^+} \{p^{(1)},p^{(2)}\ldots p^{(|d_{in}(p)-d_{out}(p)|)}\}$, 
$Q = \cup_{q\in V^-}\{q^{(1)}, q^{(2)} \ldots q^{(|d_{in}(q)-d_{out}(q)|)}\}$. We now introduce an 
edge between $p^(i) \in P$ and $q^(j) \in Q$ iff $p,q \in V$ are connected by a {\em terminal oriented
bi-directed walk} from $p$ to $q$. Let $dist^t(p,q)$ be the weight of
this walk. Then $E^b = \{(p^{(i)}, q^{(j)}) |\, dist^t(p,q) \neq \infty \,\,\wedge\,\, p,q \in V\}$. The weight
of the edge $(p^{(i)},q^{(j)})\in E^b$ is the weight of terminal oriented bi-directed walk $dist^t(p,q)$.

\begin{lemma}
\label{lem:bi-match}
A non Eulerain bi-directed graph $G(V,E)$ has a cyclic CP walk $\iff$ the balancing bi-partite graph 
$B(P,Q,E^b)$ has a perfect match.
\end{lemma}
\begin{proof}
(Forward direction:) Since $G$ has a cyclic CP walk, every un-balanced node $v \in V$(positive or negative) 
should appear at least $i \geq |d_{in}(v)-d_{out}(v)|$ times. Label each occurrence of $v$ in the cyclic CP walk by $v^{(i)}$. Also note that $\sum_{p\in V^+} |d_{in}(p)-d_{out}(p)| = \sum_{q\in V^-} |d_{in}(q) - d_{out}(q)|$, 
since $G$ has a cyclic CP walk. Now we can pair every $i^{th}$ occurrence of a positively imbalanced node $p$ to some
$j^{th}$ occurrence of a negatively imbalanced node $q$ since $p^{(i)}$ and $q^{(j)}$ are connected by a terminal oriented bi-directed
walk in the cyclic CP walk. Every such pairing corresponds to a matched edge in $B(P,Q,E^b)$.

\vspace{0.1in}

\noindent(Reverse direction:) Consider the perfect match $M_b$ in $B(P,Q,E^b)$. For every edge $(p^{(i)}, q^{(j)}) \in M_b$
over-lay the underlying oriented bi-directed walk between $p,q \in V$ on $G(V,E)$. This makes $G(V,E)$ a balanced
multi-bi-directed graph. Then by Lemma~\ref{lem:multi-lemma} we can construct a cyclic CP walk in $G$. $\Box$
\end{proof}

\subsection{Constructing a family of Eulerian multi-bi-directed graphs}
We now give a construction for generating Eulerian multi-bi-directed graphs corresponding to a given non 
Eulerian bi-directed graph which has a cyclic CP walk. We call this a {\em Balancing Match Family} denoted
by $\mathcal{F}$. Lemma~\ref{lem:bi-match} can be used to generate $\mathcal{F}$. Assume that $G(V,E)$ is a non Eulerian bi-directed graph that has a cyclic CP walk. The following construction generates a
family of Eulerian multi-bi-directed graphs corresponding to $G(V,E)$.
\begin{itemize}
\item {\sf STEP-1:} Create a balancing bi-partite graph $B(P,Q,E^m)$ corresponding to $G(V,E)$ by choosing 
some terminal oriented bi-directed walk between $p^{(i)}\in P$ and $q^{(j)}\in Q$.
\item {\sf STEP-2:} Find a perfect match $M_b$ in $B(P,Q,E^m)$. For each edge in $M_b$ overlay the 
corresponding terminal oriented bi-directed walk on $G(V,E)$. This generates a Eulerian multi-bi-directed
graph $G^m(V,E^m)$.
\end{itemize}

The following Lemma~\ref{lem:family} is easy to see.
\begin{lemma}
\label{lem:family}
If $G(V,E)$ is a non Eulerian bi-directed graph that has a cyclic CP walk, then every corresponding
Eulerian multi-bi-directed graph $G^m(V,E^m)$ belongs to the family $\mathcal{F}$.
\end{lemma}

The following Lemma gives an expression for the weight of
any $G^m(V,E^m) \in \mathcal{F}$.
\begin{lemma}
\label{lem:weight}
Let $G(V,E,c)$ be a non Eulerian weighted bi-directed graph which has a cyclic CP walk $c: E\rightarrow \mathbb{R^+}$.
Let $G^m(V,E^m,c) \in \mathcal{F} $ be some Eulerian multi-bi-directed graph. Then, 
${\mathcal W}(G^m) = \displaystyle\sum_{e\in E} c(e) + \displaystyle\sum_{(p^{(i)},q^{(j)})\in M_b}dist^t(p,q)$, where 
$M_b$ is a perfect match in $B(P,Q,E^b)$.
\end{lemma}
\begin{proof}
Since $G^m$ is Eulerian it should cover every edge in $G$ - this corresponds to the first term. Secondly,
since $G^m(V,E^m,c) \in \mathcal{F}$ the cost of multi-edges coming from overlaying the terminal bi-directed
walk corresponds to the match $M_b$ in $B(P,Q,E^b)$. This corresponds to the second term. $\Box$
\end{proof}

\subsection{An algorithm for finding an optimal cyclic CP walk}
We now put together all the results in the preceding sub-section(s) to give an algorithm to find $G^*(V,E^*)$.
The algorithm is summarized in the following steps.

\begin{itemize}
\item {\sf STEP-1:} We first identify positive and negative imbalanced nodes in $G$. 
Let $V^+ = \{v | d_{in}(v) - d_{out}(v) > 0\}$, $V^- =\{v | d_{in}(v) - d_{out}(v) < 0\}$
\item {\sf STEP-2:} Find the cost of a {\em terminal oriented shortest bi-directed walk} between every pair 
$(v,u) \in V^+\times V^-$. Let this cost be denoted as $dist^t(v,u)$.
\item {\sf STEP-3:} Create a {\em balancing bi-partite graph} $B(P,Q,E^b)$ as follows. Let $P = \displaystyle\cup_{v\in V^+} \{v^{(1)},v^{(2)},\ldots, v^{(|d_in(v)-d_out(v)|)}\}$, $Q = \displaystyle\cup_{u\in V^-} \{u^{(1)},u^{},\ldots, u^{(|d_{in}(u)-d_{out}(u)|)}\}$, $E = \{(v^{(i)},u^{(j)}) | v^{(i)} \in P \wedge u^{(j)} \in Q\}$. The cost of an edge $c(v^{(i)},u^{(j)}) = dist^t(v,u)$. 
\item {\sf STEP-4:} Find a minimum cost perfect match in $B$. Let this match be $M_b$. If $B$ does not have a
perfect match then $G$ does not have a cyclic CP walk.
\item {\sf STEP-5:} For each edge $(v^{(i)}, u^{(j)}) \in M_b$ , overlay the terminal oriented shortest bi-directed
walk between $v$ and $u$ in the $G(V,E)$. After overlaying all the terminal oriented bi-directed walks from $M_b$ on to 
$G(V,E)$ we obtain $G^*(V,E^*)$. We will prove that it is optimal in Theorem~\ref{thm:thm1}.
\end{itemize}

\begin{theorem}
\label{thm:thm1}
If $G(V,E)$ is a bi-directed graph that has a cyclic CP walk, then the cost of this cyclic CP walk is equal to 
${\cal W}(G^*) = \displaystyle\sum_{e\in E}c(e) + \displaystyle\sum_{(v^{(i)},u^{(j)})\in M_b}dist^t(v,u)$. Here
$M_b$ is the min-cost perfect match in the balancing bi-partite graph $B$.
\end{theorem}
\begin{proof}
By Lemma~\ref{lem:family} the multi-bi-directed graph $G^*(V,E^*)$ belongs to $\mathcal{F}$. Now by Lemma~\ref{lem:weight},
any optimal solution has to minimize the second term ($\sum_{(p^{(i)},q^{(j)}}dist^t(p,q)$). To minimize this the
algorithm chooses {\em shortest terminal oriented bi-directed walk} in {\sf STEP-2}. Finally, in {\sf STEP-5} the 
algorithm finds a minimum cost perfect match. Both these steps ensure that $\mathcal{W}(G^*)$ is minimum in the entire
family of multi-bi-directed graphs in $\mathcal{F}$.$\Box$
\end{proof}

\subsection{Runtime analysis of the algorithm to find a cyclic CP walk}
Let $p = \max\{|V^+|, |V^-|\}$ and $d_{max} = \displaystyle\max_{v\in V}\{|d_{in}(v)-d_{out}(v)\}$.
{\sf STEP-2} of the algorithm runs in $\Theta(p(|V|+|E|)\log(|V|))$ time to compute $dist^t(v,u)$. In 
{\sf STEP-3} $|P|\leq d_{max}p$ , $|Q|\leq d_{max}p$. For {\sf STEP-4} Hungarian method can be applied
to solve the weighted matching problem in $\Theta((d_{max}p)^3)$ time. So the total runtime of this
deterministic algorithm is $\Theta(p(|V|+|E|)\log(|V|) + (d_{max}p)^3)$. As mentioned before
if $p$ is much smaller than $|V|$ this algorithm performs better than the bi-directed flow algorithm.
\subsection{Runtime analysis of the algorithm to find SDDNA}
Since SDDNA runs on a bi-directed de Bruijn graph which is un-weighted, {\sf STEP-2} of the algorithm
runs in $\Theta(p(|V|+|E|))$ time -- because we don't need to use a Heap, we just do a BFS on the bi-directed
graph. The rest of the analysis for the runtime remains the same and the total run time of the algorithm 
is $\Theta(p(|V|+|E|)+(d_{max}p)^3)$.

\section{A $\Theta(p(|V|+|E|))$ Time Heuristic Algorithm for the SDDNA Problem}
\label{sec:gdy_algo}
From the analysis in the previous section, to solve the SDDNA problem deterministically we need to spend
$\Theta(p(|V|+|E|) + (d_{max}p)^3)$ time. However if we just replace the Hungarian method in {\sf STEP-4}
with a simple greedy algorithm we can get rid of the $(d_{max}p)^3$ term in the asymptotic complexity.
Although we have a constant $2/3-\epsilon$ approximation algorithm for {\em maximum weighted matching},
we are not aware of any constant approximation algorithms for {\em minimum weight perfect matching}.
As a result this just remains as a heuristic algorithm. On the other hand this algorithm seems to be
performing very close to the optimal (see Section~\ref{sec:gdy_heuristic}).

\section{Dealing with Practical Bi-directed de Bruijn Graphs with no Cyclic CP Walks}
\label{sec:no-cpp}
As we have mentioned earlier most of the bi-directed de Bruijn graphs constructed from the
reads do not satisfy the sufficient condition for cyclic CP walks. In such cases our algorithm
can still be used, by modifying it to find a {\em maximum} match in the balancing bi-partite graph  
rather than perfect match. We can introduce a hypothetical node $h$ and connect all the 
un-matched nodes in the balancing bi-partite graph to $h$ with appropriate bi-directed edges and
thus make all the original nodes balanced. We can now find a cyclic CP walk in this hypothetical 
graph. Every sub-walk in the cyclic CP walk that starts from $h$ and ends at $h$ can be reported
as a {\em contig}. Thus our algorithm is capable of handling cases when the bi-directed graph cannot have a cyclic CP walk.

\section{Experimental Results}
\label{sec:exp}
As we have mentioned in the previous sections the asymptotic complexity of our algorithm depends on 
$p$ -- the maximum of positively and negatively imbalanced nodes. In the case of de Bruijn
graphs $d_{max}\leq |\Sigma|$, where $|\Sigma|$ is the size the alphabet from which the strings
are drawn. In our case this is exactly four. So we can safely ignore $d_{max}$ in the case of de Bruijn
graphs and just concentrate on $p$. In the rest of the discussion we would like to refer to $p$
as the number of imbalanced nodes.
\subsection{Estimation of the mean of the random variable $\frac{p}{|V|}$}
It is clear that $p$ is a random variable with support in $[0,|V|]$. So we would like to 
estimate the expected number of imbalanced nodes in a graph with $|V|$ bi-directed edges. We
estimated the mean of the random variable $\frac{p}{|V|}$ from several samples of bi-directed
de Bruijn graphs constructed from reads from a plant genome. A simple $t-$test is applied to 
to estimate the $95\%$ confidence interval of $\frac{p}{|V|}$. See Table~\ref{tab:imbal} for
the details of the samples used. Notice that as we increase the size of $k$ (de Bruijn graph
order) from $21$ to $25$, the number of imbalanced nodes in columns corresponding to $|V^+|$
and $|V^-|$ reduces. This is because increasing $k$ reduces the number of edges which may reduce
the number of imbalanced nodes. On the other hand for a fixed value of $k$ the number of
imbalanced nodes increases consistently with the nodes. However the rate of growth is very
slow compared to the rate of growth of the number of nodes. Finally we use this evidence to hypothesize
that the number of imbalanced nodes in practical bi-directed graphs is only between $0.087\%$
to $0.133\%$ of the number of nodes in the graph, with a probability of $95\%$.
\subsection{Performance of the greedy heuristic and handling cases which do not have cyclic CP walks}
\label{sec:gdy_heuristic}
The greedy heuristic described in Section~\ref{sec:gdy_algo} has been compared with the optimal
{\em maximum match} with {\em minimum cost}. As we mentioned in Section~\ref{sec:no-cpp} many 
of these graphs do not contain cyclic CP walks so they do not have a perfect match. To cope with
this situation we treated the balancing bi-directed graph as a complete bi-directed graph, by
introducing a hypothetical edge with large cost whenever there is no edge between two nodes
in the original graph. Thus we just used the size of the match to compare the cost of the greedy
algorithm and the optimal algorithm to solve the matching problem. Table~\ref{tab:appx} gives
the details of the balancing bi-partite graph obtained from several read samples. As we have
mentioned before, to get the approximation ratio we treated a balancing bi-partite graph as a complete
bi-partite graph $K_{p,p}$. If $M_{opt}$ and $M_{gdy}$ are the sizes of {\em maximum} and {\em maximal}
matches then we treated the cost of hypothetical perfect match as $(p-|M_{opt}|)$ and $(p-|M_{gdy}|)$
and their ratio is used as approximation ratio. Finally from the evidence in Table~\ref{tab:appx}
we hypothesize that the approximation ratio for this is between $1.008$ and $1.016$ with a probability
of $95\%$.

\subsection{Implementation and Data}
An implementation of the algorithms discussed is available at \url{http://trinity.engr.uconn.edu/~vamsik/fast_cpp.tgz}.

\section{Conclusion and further research}
In this paper we have given an algorithm for cyclic Chinese Postman walk on a bi-directed de Bruijn graph. Our
algorithm is based on identifying shortest bi-directed walks and weighted matching. This algorithm performs
asymptotically better than the bi-directed flow algorithm when the number of imbalanced nodes are much smaller
than the nodes in the bi-directed graph. On the other hand this algorithm can also handle the instances of
bi-directed graphs which does not have a cyclic CP walk and provide a minimal set of walks, cyclic walks
which cover every edge in the bi-directed graph at least once. 

There are several research directions which can be pursued. Firstly, we need to address how the addition of
paired reads may impose new constraints on the cyclic CPP walk. Secondly, while Eulerization of the bi-directed graph
we have chosen the shortest path bi-directed path, however this may not correspond to the repeating region in 
the genome. Other strategy to make the graph Eulerian is to choose the path with maximum read multiplicity. 
This on other hand may increase the length of the Chinese walk, can we simultaneously optimize these two 
objectives ?.

\begin{table}
\begin{center}
\begin{tabular}{|c|c|c|c|c|c|c|c|c|}
\hline
READS &  $k$ & NODES & P-IMBAL & N-IMBAL & \multicolumn{3}{|c|}{BAL-BI-GRAPH} & \\
\hline
 &&&$|V^+|$&$|V^-|$&$|P|$ & $|Q|$ & $p$ & $\frac{p\times 100}{|V|}$ \\
\hline
102400 & 21 & 1588569 & 1157 & 1133 & 1186 & 1173 & 1186 & 0.075 \\ 
\hline
153600 & 21 & 2353171 & 2240 & 2141 & 2298 & 2211 & 2298 & 0.098 \\ 
\hline
204800 & 21 & 3097592 & 3509 & 3492 & 3601 & 3590 & 3601 & 0.116 \\ 
\hline
256000 & 21 & 3825101 & 4953 & 5004 & 5074 & 5131 & 5131 & 0.134 \\ 
\hline
307200 & 21 & 4538734 & 6719 & 6748 & 6878 & 6912 & 6912 & 0.152 \\ 
\hline
358400 & 21 & 5235821 & 8586 & 8603 & 8789 & 8802 & 8802 & 0.168 \\ 
\hline
409600 & 21 & 5917489 & 10665 & 10693 & 10914 & 10934 & 10934 & 0.185 \\ 
\hline
102400 & 25 & 1202962 & 569 & 521 & 588 & 540 & 588 & 0.049 \\ 
\hline
153600 & 25 & 1788533 & 1104 & 1026 & 1139 & 1062 & 1139 & 0.064 \\ 
\hline
204800 & 25 & 2362981 & 1744 & 1708 & 1788 & 1759 & 1788 & 0.076 \\ 
\hline
256000 & 25 & 2927656 & 2521 & 2523 & 2579 & 2592 & 2592 & 0.089 \\ 
\hline
307200 & 25 & 3484849 & 3370 & 3414 & 3451 & 3517 & 3517 & 0.101 \\ 
\hline
358400 & 25 & 4032490 & 4333 & 4369 & 4441 & 4485 & 4485 & 0.111 \\ 
\hline
409600 & 25 & 4571554 & 5390 & 5467 & 5518 & 5613 & 5613 & 0.123 \\ 
\hline
\multicolumn{9}{|c|}{}\\ 
 \multicolumn{9}{|c|}{$\left[\bar{x}-z_{\frac{\alpha}{2}}\frac{S}{\sqrt{n}}\,,\,\bar{x}+z_{-\frac{\alpha}{2}}\frac{S}{\sqrt{n}} \right]$ :  95\% C.I for average $\frac{p\times 100}{|V|}$ is $[0.0872\%,0.1330\%]$} \\ 
\hline
\end{tabular}

\end{center}
\caption{The value of $p$ on short read data from a plant genome sequencing data from CSHL}
\label{tab:imbal}
\end{table}

\begin{table}
\begin{center}
\begin{tabular}{|c|c|c|c|c|c|c|c|c|}
\hline
READS &  $k$ & NODES & $\max{|P|, |Q|}$ &\multicolumn{2}{|c|}{OPT} & \multicolumn{2}{|c|}{GRDY} & APX-RATIO\\ 
\hline
 & & & &SIZE&COST&SIZE&COST \\
\hline
 &&$|V|$&$p$&$|M_{opt}|$ & $p-|M_{opt}|$ & $|M_{gdy}|$ & $p-|M_{gdy}|$ & $\frac{GDY}{OPT}$ \\
\hline
102400 & 21 & 1202962 & 1186 & 416 & 770 & 406 & 780 & 1.0130 \\ 
\hline
153600 & 21 & 1788533 & 2298 & 725 & 1573 & 704 & 1594 & 1.0134 \\ 
\hline
204800 & 21 & 2362981 & 3601 & 1092 & 2509 & 1073 & 2528 & 1.0076 \\ 
\hline
256000 & 21 & 3825101 & 5131 & 1479 & 3652 & 1450 & 3681 & 1.0079 \\ 
\hline
307200 & 21 & 4538734 & 6912 & 1929 & 4983 & 1879 & 5033 & 1.0100 \\ 
\hline
358400 & 21 & 5235821 & 8802 & 2385 & 6417 & 2329 & 6473 & 1.0087 \\ 
\hline
409600 & 21 & 5917489 & 10934 & 2876 & 8058 & 2814 & 8120 & 1.0077 \\ 
\hline
102400 & 25 & 1202962 & 588 & 152 & 436 & 147 & 441 & 1.0115 \\ 
\hline
153600 & 25 & 1788533 & 1139 & 281 & 858 & 274 & 865 & 1.0082 \\ 
\hline
204800 & 25 & 2362981 & 1788 & 438 & 1350 & 429 & 1359 & 1.0067 \\ 
\hline
256000 & 25 & 2927656 & 2592 & 619 & 1973 & 601 & 1991 & 1.0091 \\ 
\hline
307200 & 25 & 3484849 & 3517 & 809 & 2708 & 783 & 2734 & 1.0096 \\ 
\hline
358400 & 25 & 4032490 & 4485 & 995 & 3490 & 966 & 3519 & 1.0083 \\ 
\hline
409600 & 25 & 4571554 & 5613 & 1220 & 4393 & 1175 & 4438 & 1.0102 \\ 
\hline
\multicolumn{9}{|c|}{}\\ 
 \multicolumn{9}{|c|}{$\left[\bar{x}-z_{\frac{\alpha}{2}}\frac{S}{\sqrt{n}}\,,\,\bar{x}+z_{-\frac{\alpha}{2}}\frac{S}{\sqrt{n}} \right]$ :  95\% C.I for average $\frac{APX-COST}{OPT-COST}$ is $[1.0083\%,1.0106\%]$} \\ 
\hline
\end{tabular}

\end{center}
\caption{Approximation ratio of the {\sf GDY} heuristic}
\label{tab:appx}
\end{table}

\noindent {\bf Acknowledgements.} This work has been supported in
part by the following grants: NSF 0326155, NSF 0829916 and NIH
1R01GM079689-01A1.

\section*{Appendix}

\begin{table}
\begin{center}
\begin{tabular}{|llllll|}
\hline
\multicolumn{6}{|c|}{\text{iteration:$1$ , $\min\equiv dist^-[2]=1$}}\\
\hline
&$dist^+[2]=\infty$ &$dist^+[3]=\infty$ &$dist^+[4]=\infty$ &$dist^+[5]=\infty$ &$dist^+[6]=\infty$ \\
\hline 
&$dist^-[2]=1$ &$dist^-[3]=\infty$ &$dist^-[4]=\infty$ &$dist^-[5]=\infty$ &$dist^-[6]=\infty$ \\
\hline
\hline
\multicolumn{6}{|c|}{\text{iteration:$2$ , $\min\equiv dist^-[5]=2$}}\\
\hline
&$dist^+[2]=\infty$ &$dist^+[3]=\infty$ &$dist^+[4]=\infty$ &$dist^+[5]=\infty$ &$dist^+[6]=\infty$ \\
\hline 
&$dist^-[2]=1$ &$dist^-[3]=2$ &$dist^-[4]=\infty$ &$dist^-[5]=2$ &$dist^-[6]=\infty$ \\
\hline
\hline
\multicolumn{6}{|c|}{\text{iteration:$3$ , $\min\equiv dist^-[3]=2$}}\\
\hline
&$dist^+[2]=\infty$ &$dist^+[3]=3$ &$dist^+[4]=\infty$ &$dist^+[5]=\infty$ &$dist^+[6]=\infty$ \\
\hline 
&$dist^-[2]=1$ &$dist^-[3]=2$ &$dist^-[4]=\infty$ &$dist^-[5]=2$ &$dist^-[6]=\infty$ \\
\hline
\hline
\multicolumn{6}{|c|}{\text{iteration:$4$ , $\min\equiv dist^-[6]=3$}}\\
\hline
&$dist^+[2]=\infty$ &$dist^+[3]=3$ &$dist^+[4]=\infty$ &$dist^+[5]=3$ &$dist^+[6]=\infty$ \\
\hline 
&$dist^-[2]=1$ &$dist^-[3]=2$ &$dist^-[4]=\infty$ &$dist^-[5]=2$ &$dist^-[6]=3$ \\
\hline
\hline
\multicolumn{6}{|c|}{\text{iteration:$5$ , $\min\equiv dist^-[5]=3$}}\\
\hline
&$dist^+[2]=\infty$ &$dist^+[3]=3$ &$dist^+[4]=\infty$ &$dist^+[5]=3$ &$dist^+[6]=\infty$ \\
\hline 
&$dist^-[2]=1$ &$dist^-[3]=2$ &$dist^-[4]=4$ &$dist^-[5]=2$ &$dist^-[6]=3$ \\
\hline
\hline
\multicolumn{6}{|c|}{\text{iteration:$6$ , $\min\equiv dist^-[3]=3$}}\\
\hline
&$dist^+[2]=4$ &$dist^+[3]=3$ &$dist^+[4]=\infty$ &$dist^+[5]=3$ &$dist^+[6]=\infty$ \\
\hline 
&$dist^-[2]=1$ &$dist^-[3]=2$ &$dist^-[4]=4$ &$dist^-[5]=2$ &$dist^-[6]=3$ \\
\hline
\hline
\multicolumn{6}{|c|}{\text{iteration:$7$ , $\min\equiv dist^-[4]=4$}}\\
\hline
&$dist^+[2]=4$ &$dist^+[3]=3$ &$dist^+[4]=\infty$ &$dist^+[5]=3$ &$dist^+[6]=\infty$ \\
\hline 
&$dist^-[2]=1$ &$dist^-[3]=2$ &$dist^-[4]=4$ &$dist^-[5]=2$ &$dist^-[6]=3$ \\
\hline
\hline
\multicolumn{6}{|c|}{\text{\sf The shortest path from node $1$ to node $4$ is of length $4$}} \\
\hline
\end{tabular}

\end{center}
\caption{Changes in the $dist^+$ and $dist^-$ labels in every iteration of Algorithm~\ref{algo:shortest_path} on
the bi-directed graph in Figure~\ref{fig:bi-walk-example}}
\label{tab:bi-example}
\end{table}

\bibliographystyle{splncs}

\end{document}